%% file: main.tex
\documentclass[conference]{IEEEtran}
\IEEEoverridecommandlockouts

\input{macros}

\begin{document}

\title{Equivalence of Insertion/Deletion Correcting Codes for $d$-dimensional Arrays 
}

\author{%
	\IEEEauthorblockN{
	    {\textbf{Evagoras Stylianou}},
				{\textbf{Lorenz Welter}},
				{\textbf{Rawad Bitar}},
                {\textbf{Antonia Wachter-Zeh}},
		and {\textbf{Eitan Yaakobi}}
	} 
\thanks{ES, LW, RB and AW-Z are with the ECE department at the Technical University of Munich. EY is with the CS department of Technion --- Israel Institute of Technology. Emails: \{evagoras.stylianou, lorenz.welter, rawad.bitar, antonia.wachter-zeh\}@tum.de, yaakobi@cs.technion.ac.il.}
\thanks{This project has received funding from the European Research Council (ERC) under the European Union’s Horizon 2020 research and innovation programme (grant agreement No. 801434), from the Technical University of Munich - Institute for Advanced Studies, funded by the German Excellence Initiative and European Union Seventh Framework Programme under Grant Agreement No. 291763.}
}

 \IEEEaftertitletext{\vspace{-2ex}}

\maketitle

\begin{abstract}
 We consider the problem of correcting insertion and deletion errors in the $\di$-dimensional space. This problem is well understood for vectors (one-dimensional space) and was recently studied for arrays (two-dimensional space). For vectors and arrays, the problem is motivated by several practical applications such as DNA-based storage and racetrack memories. From a theoretical perspective, it is interesting to know whether the same properties of insertion/deletion correcting codes generalize to the $\di$-dimensional space. In this work, we show that the equivalence between insertion and deletion correcting codes generalizes to the $\di$-dimensional space. As a particular result, we show the following missing equivalence for arrays: a code that can correct $t_\mathrm{r}$ and $t_\mathrm{c}$ row/column deletions can correct any combination of $t_\mathrm{r}^{\mathrm{ins}}+t_\mathrm{r}^{\mathrm{del}}=t_\mathrm{r}$ and $t_\mathrm{c}^{\mathrm{ins}}+t_\mathrm{c}^{\mathrm{del}}=t_\mathrm{c}$ row/column insertions and deletions. The fundamental limit on the redundancy and a construction of insertion/deletion correcting codes in the $\di$-dimensional space remain open for future work.
\end{abstract}

\IEEEpeerreviewmaketitle

\section{Introduction}
\input{introduction.tex}

\section{Notation and Preliminaries}
\input{preliminaries-notations}

\section{Symmetric Insertion/Deletion Equivalence}
\input{t-symmetric}

\section{Equivalence of insertion and deletions correcting codes: general case}
\input{t-general}

\section{Insdel Equivalence}
\input{insdel-equiv}


\newpage
\bibliographystyle{ieeetr}
\bibliography{bibliography}

\appendix
\input{appendix}\label{sec:appendix}
\end{document}

%% file: macros.tex
\usepackage{amsmath,amssymb,amsfonts}
\usepackage{algorithmic}
\usepackage{graphicx}
\usepackage{textcomp}
\usepackage{xcolor}
\def\BibTeX{{\rm B\kern-.05em{\sc i\kern-.025em b}\kern-.08em T\kern-.1667em\lower.7ex\hbox{E}\kern-.125emX}}

\usepackage{tikz}
\usepackage{bbm}
\usepackage[capitalize]{cleveref}

\usetikzlibrary{quotes,arrows.meta}
\usepackage{balance}

\usepackage{amssymb}
\usepackage{amsmath}

\usepackage{amsthm}

\usepackage{cite}
\usepackage{graphicx}
\usepackage{wrapfig}
\usepackage{caption}
\usepackage{subfig}

\usepackage{array}
\usepackage{multirow}
\usepackage{arydshln}
\usepackage{color}
\usepackage{xcolor}

\usepackage{tcolorbox}
\usepackage[framemethod=tikz]{mdframed}

\usepackage[format=plain,
            font=it]{caption}

\crefname{claim}{Claim}{Claims}

\newcommand{\tikznode}[2]{%
	\ifmmode%
	\tikz[remember picture,baseline=(#1.base),inner sep=0pt] \node (#1) {$#2$};
	\else
	\tikz[remember picture,baseline=(#1.base),inner sep=0pt] \node (#1) {#2};%
	\fi}

\definecolor{darkgreen}{rgb}{0,.39,0}

\definecolor{darkred}{rgb}{0.4,0,0}

\newtheorem{theorem}{\bf Theorem}
\newtheorem{lemma}{\bf Lemma}
\newtheorem{claim}{\bf Claim}
\newtheorem{corollary}{\bf Corollary}
\newtheorem{counterexample}{\bf Counterexample}

\definecolor{dnared}{rgb}{1,.27,0}  
\definecolor{dnablue}{rgb}{0,1,1}
\definecolor{dnapurple}{rgb}{0.82,0.13,0.56}
\definecolor{dnagreen}{rgb}{0,1,0}

\usepackage{tikz}

\usetikzlibrary{%
	arrows,%
	automata,%
	backgrounds,%
	calc,%
	patterns,%
	plotmarks,%
	positioning,%
	shapes.geometric,%
	shapes,%
	snakes,%
	decorations.pathmorphing,%
	decorations.shapes,%
	graphs,%
	chains,%
	arrows.meta, %
}

\usepackage{verbatim}
\usepackage{tabularx}
\usepackage{lipsum}
\usepackage{stackengine}
\usepackage{nicefrac}

\newcommand{\bfc}{{\boldsymbol c}}

\newcommand{\bfe}{{\boldsymbol e}}

\newcommand{\bfo}{{\boldsymbol o}}

\newcommand{\bfr}{{\boldsymbol r}}

\newcommand{\bft}{{\boldsymbol t}}
\newcommand{\bfu}{{\boldsymbol u}}

\newcommand{\bfB}{{\mathbf B}}
\newcommand{\bfC}{{\mathbf C}}
\newcommand{\bfD}{{\mathbf D}}

\newcommand{\bfF}{{\mathbf F}}
\newcommand{\bfG}{{\mathbf G}}

\newcommand{\bfI}{{\mathbf I}}

\newcommand{\bfX}{{\mathbf X}}
\newcommand{\bfY}{{\mathbf Y}}

\usetikzlibrary{matrix,patterns}
\usetikzlibrary{fit}

\newcommand{\cC}{{\cal C}}

\newcommand{\cI}{{\cal I}}

\newcommand{\cO}{{\cal O}}
\newcommand{\cP}{{\cal P}}

\newcommand{\dball}{\ensuremath{\mathbb{D}}}
\newcommand{\iball}{\ensuremath{\mathbb{I}}}
\newcommand{\idball}{\ensuremath{\mathbb{ID}}}

\definecolor{bostonuniversityred}{rgb}{0.8, 0.0, 0.0}

\usepackage{xcolor}

\newcommand{\x}{\mathrm{x}}

\usepackage{mathdots}

\newcommand{\bfXov}{\overline{\mathbf{X}}}
\usepackage[utf8]{inputenc}

\newcommand{\tensor}{array}
\newcommand{\tensors}{arrays}
\newcommand{\di}{\ensuremath{d}}
\newcommand{\cart}{\ensuremath{\bigotimes}}
\newcommand{\odi}{\ensuremath{\otimes\di}}
\newcommand{\tcode}{\ensuremath{\bft}}
\newcommand{\proj}[1]{\ensuremath{\mathcal{P}_{#1}}}
\newcommand{\invproj}[1]{\ensuremath{\mathcal{P}_{#1}^{-1}}}
\newcommand{\bfone}{\ensuremath{\mathbf{1}}}
\newcommand{\allone}{\ensuremath{\bfone^{(\di)}}}
\newcommand{\allv}[1]{\ensuremath{#1 \allone}}
\newcommand{\insvec}{\ensuremath{\bfu}}
\newcommand{\insentry}{\ensuremath{u}}


%% file: introduction.tex
Coding for insertions and deletions received a lot of attention due to new applications such as DNA-based data storage \cite{heckel2019characterization,buschmann2013levenshtein}, synchronization errors \cite{helberg1993coding,sala2016synchronizing} and racetrack memories \cite{chee2018coding}. An important notion in this class of codes is the equivalence of insertion and deletion errors. In his original work \cite{levenshtein1966binary}, Levenshtein showed that a code can correct $t$ deletions in a length-$n$ vector if and only if it can correct any combination of $t_{\mathrm{i}}$ insertions and $t_\mathrm{d}$ deletions such that $t_\mathrm{i}+t_\mathrm{d} = t$. A more intuitive proof of the equivalence, which line of thoughts we follow in this work, is given in \cite{Cullina2014}. A code $\cC$ correcting deletions in $q$-ary length-$n$ vectors is evaluated by its redundancy defined as $R \triangleq n- \log_q |\cC|$. The redundancy of $t$-deletion-correcting codes is bounded from below by $t\log_q n - \cO(1)$ \cite{levenshtein1966binary,Cullina2014}. The asymptotical tightness of this bound is shown using the Varshamov-Tenengolts codes \cite{levenshtein1966binary,VarshTene-SingleDeletion1965,tenengolts1984nonbinary} that can correct one deletion. Several recent works considered constructing binary $t$-deletion-correcting codes, $t>1$, whose redundancy approach the previously mentioned lower bound \cite{GuruswamiWang-HighNoiseHighRateDeletions_2017,brakensiek2017efficient,hanna2018guess,Gabrys-TwoDeletions_2018,Sima-TwoDeletions_2019,SimaBruck-kDeletions_2020,guruswami2020explicit,sima2020systematictdel}.

Codes correcting insertions and deletions in two-dimensional arrays have been investigated in \cite{krishnamurthy2019trace,bakirtas2021database,bitar2021criss,chee2021two,hagiwara2020conversion,welter2021multiple}. The model considered in \cite{bitar2021criss,chee2021two,hagiwara2020conversion,welter2021multiple} is that of coding for row/column insertions and deletions in two-dimensional arrays. In \cite{hagiwara2020conversion}, Hagiwara constructed codes that can correct up to $t_c$ column and $t_r$ row deletions where $t_r$ and $t_c$ are predetermined. In \cite{bitar2021criss,welter2021multiple}, the authors constructed codes correcting a variable number of column and row deletions for a predetermined number of total deletions. In addition, they provided a lower bound on the redundancy of codes correcting insertions and deletions in arrays. Moreover, they generalized the equivalence between insertions and deletions across each dimension (columns and rows), separately. More precisely, the authors showed that given an integer $t$, a code can correct $t_r$ and $t_c$, for all $t_r+t_c = t$, row and column deletions if and only if it can correct the same number of rows/columns insertions. However, combinations of insertions and deletions of columns (and rows) was not studied.

In this work we generalize the equivalence between codes correcting insertions and deletions to the $\di$-dimensional space. In this setting, the insertions and deletions are defined as $(\di-1)$-dimensional hyperplane insertions/deletions in a $\di$-dimensional \tensor. In the $\di$-dimensional space there are $\di = \binom{\di}{\di-1}$ different types of $(\di-1)$-hyperplane deletions/insertions. Each type of deletion is indexed by the missing dimension. More precisely, let $(x_1,\dots,x_\di)$ describe the axes of the $\di$-dimensional space. Deleting a $(\di-1)$-dimensional hyperplane not containing the axis $x_i$ is referred to as an $x_i$-deletion. See~\cref{fig:planes2} for an illustrative example for $\di=3$. For a vector $\bft = (t_1,\dots,t_\di)$, a $\bft$-deletion refers to the combination of $t_i$ $x_i$-deletions for $i\in \{1,\dots,\di\}$. We show that a code can correct $\bft$-deletions if and only if it can correct $\bft$-insertions. We extend this result to combinations of insertion and deletions, i.e., we show that a code can correct $\bft$-deletions if and only if it can correct any combination of $\bft^{\mathrm{del}}$-deletions and $\bft^{\mathrm{ins}}$-insertions such that $\bft^{\mathrm{del}}+\bft^{\mathrm{ins}} = \bft$. We show that the number of $x_i$-errors (insertions plus deletions) must remain the same for the equivalence to hold.

%% file: preliminaries-notations.tex
Denote the $q$-ary alphabet by $\Sigma_q\triangleq \{0, \dots, q-1\}$ and the set of integers $\{1, \dots , n\}$  by $[n]$. Moreover, denote the set of $\di$-dimensional \tensors, in short called \tensors, by $\Sigma_q^{\cart_{i=1}^{\di} n_i} = \Sigma_q^{n_1 \times \dots \times n_\di} \triangleq \Sigma_q^{n_1} \times \dots \times \Sigma_q^{n_\di}$ with entries in $\Sigma_q$. We abbreviate $\Sigma_q^{n^{\odi}} \triangleq \Sigma_q^{\cart_{i=1}^{\di} n}$, if $n_i=n_j$ for all $i,j\in [\di]$. Let $(x_1,\dots,x_\di)$ describe the axes of the $\di$-dimensional space. For two-dimensional arrays an $x_1$-deletion corresponds to a column deletion and an $x_2$-deletion to a row deletion. See \cref{fig:planes2} for an illustration in the $3$-dimensional space.
\begin{figure}[t]
    \centering
    \input{Figures/Intro_fig}
    \vspace{-0.5ex}
    \caption{Illustration of all possible plane deletion or insertion in a $3$-dimensional \tensor. \vspace{-0.78cm}}
\label{fig:planes2}
\end{figure}
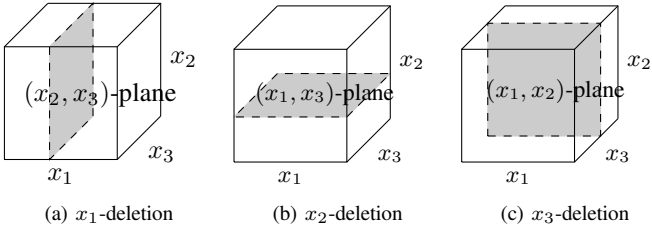
A $\tcode$-deletion where $\tcode \in \mathbb{Z}_{\geq 0}^{\di}$ corresponds to the combination of $t_i$ $x_i$-deletions for $i \in [\di]$, resulting in an \tensor\ $\widetilde{\bfX} \in \Sigma_q^{\cart_{i=1}^\di (n-t_i)}$. 
Moreover, a $\tcode$-insdel where $\tcode = \tcode^{\mathrm{ins}} + \tcode^{\mathrm{del}}$ corresponds to the combination of $\tcode^{\mathrm{del}}$-deletions and $\tcode^{\mathrm{ins}}$-insertions resulting in an \tensor\ $\widetilde{\bfX} \in \Sigma_q^{\cart_{i=1}^\di (n+(t_i^{\mathrm{ins}}-t_i^{\mathrm{del}}))}$.

For $\bfX\in \Sigma_q^{n^{\odi}}$ and $\tcode^{(\di)} \in \mathbb{Z}_{\geq 0}^{\di}$, the set of \tensors\ resulting from a $\tcode^{(d)}$-deletion in $\bfX$ is called the deletion ``ball'' and is denoted by $\dball^\di_{\tcode}(\bfX)$. We define a $\tcode^{(\di)}$-deletion correcting code $\cC \subseteq \Sigma_q^{n^{\odi}}$ as the code that can correct any $\tcode^{(\di)}$-deletion for all $\bfX\in \cC$. The all-zero vector with ``1'' in the $i$-th position is denoted by $\bfe_i$. The $\mathbf{1}^{(\di)}$ denotes the all-one vector of length $\di$. Vectors $\tcode$ of the form $\tcode = (t,\dots,t)$ are denoted by $\allv{t}$. For such vectors we denote the deletion ball by $\dball^\di_{t \bfone}(\bfX)$. The $\tcode$-insertion and the insertion balls $\iball^\di_{\tcode}(\bfX)$ and $\iball^\di_{t \bfone}(\bfX)$ are defined similarly. Moreover, the set of \tensors\ resulting from $\tcode^{(d)}$-insdel in $\bfX$ is called the insertion-deletion ``ball'' and denoted by $\idball_\tcode^\di(\bfX)$.   We define a $\tcode^{(\di)}$-insdel correcting code $\cC \subseteq \Sigma_q^{n^{\odi}}$ as the code that can correct any $\tcode^{(\di)}$-insdel for all $\bfX\in \cC$.
For an integer $t\geq 0$, a $t^{(\di)}$-deletion refers to the collection of all possible $\tcode$-deletions such that $\sum_{i=1}^{\di} t_i = t$. 
We define a $t^{(\di)}$-deletion correcting code $\cC \subseteq \Sigma_q^{n^{\odi}}$ as the code that can correct any $t^{(\di)}$-deletion for all $\bfX\in \cC$. The same notation is used for insertions.  For an integer $a$, we define $\delta_{a}(x)$ to be equal to one if $x=a$ and zero otherwise.

For $j\in [\di]$, the projection $\cP_j$ projects an \tensor\ $\bfX\in \Sigma_q^{\cart_{i=1}^\di n_i}$ along the $x_j$-th axis to an \tensor\ $\cP_j(\bfX) \in \Sigma_{q^{n_j}}^{\cart_{i\in[\di]}^{i\neq j}n_i}$. The projection $\cP_j$ preserves the order of the axes, i.e., it projects $\bfX$ from the space with axes $(x_1,\ldots,x_d)$ onto the space with axes $(x_1,\ldots,x_{j-1},x_{j+1},\ldots,x_d)$.
Moreover, we denote by $\mathcal{P}_j^{-1}$ the inverse projection, or the expansion, of an \tensor\ $\bfX\in \Sigma_{q^{n_j}}^{\cart_{i\in[\di]}^{i\neq j}n_i}$ along the $x_j$-th axis to obtain an \tensor\ $\cP_j^{-1}(\bfX)\in \Sigma_q^{\cart_{i=1}^\di n_i}$. The inverse projection $\cP_{j}^{-1}$ also preserves the order of the dimensions $(x_1,\ldots,x_d)$. 
For example, given an \tensor\ $\bfX\in\Sigma_{q^{n}}^{n\times n}$ in the $(x_1,x_3)$ space, the inverse projection $\cP_2^{-1}(\bfX)$ expands each entry of $\bfX$ to a vector in $\Sigma_{q}^n$ along $x_2$ to obtain $\cP_2^{-1}(\bfX) \in \Sigma_q^{n\times n \times n}$.

Next we state in our notation two preliminary results derived in~\cite{welter2021multiple,bitar2021criss} for the $2$-dimensional case. \cref{lem:2} is used as a building block of our proofs.
\begin{theorem}{\cite[Theorem 1]{welter2021multiple}} \label{thm:twodim} A code $\mathcal{C}\subseteq \Sigma_q^{n\times n}$ is a $t^{(2)}$-deletion correcting code if and only if it is a $t^{(2)}$-insertion correcting code, i.e., for any \tensors\ $\bfX,\bfY \in \Sigma_q^{n\times n}$,\vspace{-0.1cm}
\begin{align}
\mathbb{D}_{\tcode}^2(\mathbf{X}) \cap \mathbb{D}_{\tcode}^2(\mathbf{Y}) \neq \emptyset\;\;\text{if and only if}\;\;\mathbb{I}_{\tcode}^2(\mathbf{X}) \cap \mathbb{I}_{\tcode}^2(\mathbf{Y}) \neq \emptyset. \nonumber
\end{align}
for any choice of $\tcode \in \mathbb{Z}_{\geq0}^2$ such that $t_1+t_2 = t$.
\end{theorem}
\begin{lemma}{\cite{bitar2021criss}}\label{lem:2}
For a positive integer $m$ and $i,j\in [2]$, any two \tensors\ $\bfX \in \Sigma_q^{(m+\delta_{1}(i))\times (m+\delta_{2}(i))}$ and $\bfY \in \Sigma_q^{(m+\delta_{1}(j))\times (m+\delta_{2}(j))}$, it holds that\vspace{-0.1cm}
\begin{align*}
\mathbb{D}_{\bfe_i}^2(\mathbf{X}) \cap \mathbb{D}_{\bfe_j}^2(\mathbf{Y}) \neq \emptyset \Leftrightarrow \mathbb{I}_{\bfe_j}^2(\mathbf{X}) \cap \mathbb{I}_{\bfe_i}^2(\mathbf{Y}) \neq \emptyset.
\end{align*}
\end{lemma}

%% file: Figures/Intro_fig.tex
    \subfloat[$x_1$-deletion\label{subfig-1:dummy}]{%
         \begin{tikzpicture}
            \def\sliceZ{0.6}
            \def\side{1.5}
        
              \filldraw[color=gray!40] (\sliceZ,0,0) -- (\sliceZ,\side,0) -- (\sliceZ,\side,\side) -- (\sliceZ,0,\side)-- cycle;
            \draw[dashed]  (\sliceZ,0,0) -- (\sliceZ,\side,0) -- (\sliceZ,\side,\side) -- (\sliceZ,0,\side)-- cycle;
        
            \draw (\side,0,0) -- (\side,\side,0) node[midway,right] {$x_2$} -- (0,\side,0);
            \draw (0,0,\side) -- (\side,0,\side) node[midway,below] {$x_1$} -- (\side,\side,\side) -- (0,\side,\side) -- (0,0,\side);
            \draw (\side,0,0) -- (\side,0,\side) node[midway,below right] {$x_3$};
            \draw (\side,\side,0) -- (\side,\side,\side);
            \draw (0,\side,0) -- (0,\side,\side);
        
            \node at (1,\sliceZ,0.8) {$(x_2,x_3)$-plane};
        \end{tikzpicture}
    }
    \hfill
    \subfloat[$x_2$-deletion\label{subfig-3:dummy}]{%
         \centering
         \begin{tikzpicture}
            \def\sliceZ{0.6}
            \def\side{1.5}
            \filldraw[color=gray!40] (0,\sliceZ,0) -- (0,\sliceZ,\side) -- (\side,\sliceZ,\side) -- (\side,\sliceZ,0) -- cycle;
            \draw[dashed] (0,\sliceZ,0) -- (0,\sliceZ,\side) -- (\side,\sliceZ,\side) -- (\side,\sliceZ,0) -- cycle;
            \draw (\side,0,0) -- (\side,\side,0) node[midway,right] {\small$x_2$} -- (0,\side,0);
            \draw (0,0,\side) -- (\side,0,\side) node[midway,below] {\small$x_1$} -- (\side,\side,\side) -- (0,\side,\side) -- (0,0,\side);
            \draw (\side,0,0) -- (\side,0,\side) node[midway,below right] {\small$x_3$};
            \draw (\side,\side,0) -- (\side,\side,\side);
            \draw (0,\side,0) -- (0,\side,\side);
        
            \node at (0.9,\sliceZ,0.7) {\small $(x_1,x_3)$-plane};
        \end{tikzpicture}
    }
        \hfill
    \subfloat[$x_3$-deletion\label{subfig-2:dummy}]{%
         \centering
         \begin{tikzpicture}
            \def\sliceZ{0.6}
            \def\side{1.5}
              \filldraw[color=gray!40] (0,0,\sliceZ) -- (0,\side,\sliceZ) -- (\side,\side,\sliceZ) -- (\side,0,\sliceZ)-- cycle;
            \draw[dashed]  (0,0,\sliceZ) -- (0,\side,\sliceZ) -- (\side,\side,\sliceZ) -- (\side,0,\sliceZ)-- cycle;
        
            \draw (\side,0,0) -- (\side,\side,0) node[midway,right] {\small$x_2$} -- (0,\side,0);
            \draw (0,0,\side) -- (\side,0,\side) node[midway,below] {\small$x_1$} -- (\side,\side,\side) -- (0,\side,\side) -- (0,0,\side);
            \draw (\side,0,0) -- (\side,0,\side) node[midway,below right] {\small$x_3$};
            \draw (\side,\side,0) -- (\side,\side,\side);
            \draw (0,\side,0) -- (0,\side,\side);
        
            \node at (0.9,0.6,\sliceZ) {\small$(x_1,x_2)$-plane};
        \end{tikzpicture}
        }

%% file: t-symmetric.tex
In this section we prove the following theorem. 

\begin{theorem} \label{thm:eqinsdel3d2} A code $\mathcal{C} \subseteq \Sigma_q^{n^{\odi}}$ is a $t\bfone^{(d)}$-deletion-correcting code if and only if it is a $t\bfone^{(d)}$-insertion-correcting code.
\end{theorem}

To prove~\cref{thm:eqinsdel3d2} we need three intermediate results. In~\cref{claim:dim-proj}, we show that $\bft^{(\di)}$-deletions and $\bft^{(\di)}$-insertions in an \tensor\ $\bfX$ are not affected by the projection $\proj{\kappa}(\bfX)$ and the inverse projection $\invproj{\kappa}(\bfX) $ such that $t_\kappa = 0$. We then extend~\cref{lem:2} to the $\di$-dimensional space, cf.,~\cref{lm:dd2}, and use it as a building block in our proofs. In particular, we use~\cref{lm:dd2} to prove~\cref{thm:eqinsdel3d} showing that a code is a $\allone$-deletion-correcting code if and only if it is a $\allone$-insertion-correcting code.
Having the aforementioned results, proving~\cref{thm:eqinsdel3d2} follows by showing that for any $\mathbf{X},\mathbf{Y}\in \Sigma_q^{n^{\odi}}$, $\mathbb{D}^d_{t\bfone}(\mathbf{X}) \cap \mathbb{D}^d_{t\bfone}(\mathbf{Y}) \neq \emptyset$ if and only if $\mathbb{I}^d_{t\bfone}(\mathbf{X}) \cap \mathbb{I}^d_{t\bfone}(\mathbf{Y}) \neq \emptyset$.
The proof holds by using the exact same steps as in the proof of~\cite[Corollary~2]{bitar2021criss}, but extended to the $d$-dimensional space and is given in Appendix~\ref{sec:appendix}.

We start with the first intermediate result. 
\begin{claim} \label{claim:dim-proj} 
For any two vectors $\bfr_1,\bfr_2 \in \mathbb{N}^{\di}$ such that there exists a $\kappa \in [\di]$ for which $r_{1,\kappa} = r_{2,\kappa} = 0$ and any two \tensors\ $\bfX \in \Sigma_{q}^{\cart_{i=1}^{d}(n+r_{1,i})}$,  $\bfY \in \Sigma_{q}^{\cart_{i=1}^{d}(n+r_{2,i})}$, it holds that,\vspace{-0.08cm}
\begin{align*}
&\dball_{\bfr_1}^{\di}(\bfX) \cap \dball_{\bfr_2}^{\di}(\bfY) \neq \emptyset \Leftrightarrow \\
&\hspace{1cm}\dball_{\proj{\kappa}(\bfr_1)}^{\di-1}(\proj{\kappa}(\bfX)) \cap \dball_{\proj{\kappa}(\bfr_2)}^{\di-1}(\proj{\kappa}(\bfY)) \neq \emptyset,
\end{align*}
where $\kappa$ denotes the $\kappa$-th dimension in the $\di$-dimensional space and $\proj{\kappa}(\bfr_j) \in \mathbb{N}^{\di-1}$ is equal to $\bfr_j$ with the zero deleted in the $\kappa$-th position for $j=1,2$. \\
The same statement holds for the insertion case.
\end{claim}
\begin{proof} 
We first prove the ``if'' part. Let $\bfD \in \dball_{\bfr_1}^{\di}(\bfX) \cap \dball_{\bfr_2}^{\di}(\bfY)$ and $\bfD' \in \dball_{\proj{\kappa}(\bfr_1)}^{\di-1}(\proj{\kappa}(\bfX)) \cap \dball_{\proj{\kappa}(\bfr_2)}^{\di-1}(\proj{\kappa}(\bfY))$. The $\kappa$-th dimension is not affected by a deletion in both \tensors\ $\bfX$ and $\bfY$. Therefore, the deletions do not affect the mapping of the $q$-ary symbols to $q^n$-ary symbols along the axis $x_\kappa$, when using the projection function. Thus, the $(\di-1)$-dimensional hyperplane deletions in $\bfX, \bfY$ correspond to $(\di-2)$-dimensional hyperplane deletions in the respective projected \tensors. Hence, we have $\invproj{\kappa}(\bfD') = \bfD$. 

We now prove the ``only if'' part. By expanding the $q^n$-ary symbols to $q$-ary symbols along the $x_\kappa$-th axis, i.e., by applying the inverse projection, the $(\di-2)$-dimensional hyperplane deletions in $\proj{\kappa}(\bfX), \proj{\kappa}(\bfY)$ transform to ${(\di-1)}$-dimensional hyperplane deletions in $\bfX, \bfY$ with no $x_\kappa$-deletions. This follows from the definition of the projections.
\end{proof}

We now state and prove the second intermediate result.
\begin{lemma}\label{lm:dd2} For positive integers $m_1,\dots,m_\di$ and $i,j \in [\di]$, for any two \tensors\ $\bfX \in \Sigma_q^{\cart_{\ell=1}^\di (m_\ell+\delta_\ell(i))}$ and $\bfY \in \Sigma_q^{\cart_{\ell=1}^\di (m_\ell+\delta_\ell(j))}$ it holds that,
\begin{align*}
\dball_{\bfe_i}^\di(\mathbf{X}) \cap \dball_{\bfe_j}^\di(\bfY) \neq \emptyset \Leftrightarrow \iball_{\bfe_j}^\di(\mathbf{X}) \cap \iball_{\bfe_i}^\di(\bfY) \neq \emptyset. 
\end{align*}
\end{lemma}
\begin{proof}
We only show the ``if'' part. The ``only if'' part is proven similarly. We prove the statement by induction over the dimensions. The two-dimensional case, i.e., $\di = 2$, was already shown in \cite{bitar2021criss} and is recalled in Lemma \ref{lem:2}. To illustrate the proof techniques used in the proof and in this work, we choose the three-dimensional case as the base case of the induction.
Without loss of generality, we show that
\begin{align*}
    \dball^\di_{\bfe_1}(\bfX)\cap \dball^\di_{\bfe_2}(\bfY) \neq\emptyset &\Leftrightarrow \iball^\di_{\bfe_2}(\bfX)\cap \iball^\di_{\bfe_1}(\bfY) \neq\emptyset. 
\end{align*}

\emph{Base case $\di = 3$:}  We show that
\begin{align*}
    \dball^3_{\bfe_1}(\bfX)\cap \dball^3_{\bfe_2}(\bfY) \neq\emptyset &\Leftrightarrow \iball^3_{\bfe_2}(\bfX)\cap \iball^3_{\bfe_1}(\bfY) \neq\emptyset. 
\end{align*}
For $\bfX \in \Sigma_q^{(n+1) \times n \times n}$ and $\bfY \in \Sigma_q^{n \times (n+1) \times n}$ let $\bfD \in \dball^3_{\bfe_1}(\bfX) \cap \dball^3_{\bfe_2}(\bfY)$. Since the deletion does not affect both arrays along the axis $x_3$, then we can project along this axis to transform the given three-dimensional deletion problem to a two-dimensional deletion problem by \cref{claim:dim-proj}. Thus, the $\bfe_1$-deletion in $\bfX$ converts to a row deletion in $\proj{3}(\bfX)$ and the $\bfe_2$-deletion in $\bfY$ to a column deletion in $\proj{3}(\bfY)$. Hence, it holds that $\proj{3}(\bfD) \in \dball^2_{\bfe_1}(\proj{3}(\bfX)) \cap \dball^2_{\bfe_2}(\proj{3}(\bfY))$. By \cref{lem:2}, we have the following statement
\begin{align*}
    &\dball^2_{\bfe_1}(\proj{3}(\bfX))\cap \dball^2_{\bfe_2}(\proj{3}(\bfY)) \neq \emptyset  \\
    &\Leftrightarrow \iball^2_{\bfe_2}(\proj{3}(\bfX))\cap \iball^2_{\bfe_1}(\proj{3}(\bfY)) \neq \emptyset.
\end{align*}
Therefore, there exists a $\proj{3}(\bfI) \in \iball^2_{\bfe_2}(\proj{3}(\bfX)) \cap \iball^2_{\bfe_1}(\proj{3}(\bfY))$. Let $\bfI=\invproj{3}(\proj{3}(\bfI))$, by \cref{claim:dim-proj} the previous statement is equivalent to stating that there exists a $\bfI \in \iball^3_{\bfe_2}(\bfX) \cap \iball^3_{\bfe_1}(\bfY)$. This results from applying the inverse projection $\invproj{3}(\cdot)$ on the respective arrays, transforming the row/column insertions in the two-dimensional space to $\bfe_1$-/$\bfe_2$-insertion in the three-dimensional space; thus concluding the base case.

\emph{Induction hypothesis:} For a positive integer $\di-1$ assume that it holds that
\begin{align*}
\dball_{\bfe_1}^{\di-1}(\bfX) \cap \dball_{\bfe_2}^{\di-1}(\bfY) \neq \emptyset\Leftrightarrow\iball_{\bfe_2}^{\di-1}(\bfX) \cap \iball_{\bfe_1}^{\di-1}(\bfY) \neq \emptyset.
\end{align*}
\emph{Induction step:} Given the induction hypothesis we show that the equivalence holds also for $\di$, i.e., 
\begin{align*}
\dball_{\bfe_1}^{\di}(\bfX) \cap \dball_{\bfe_2}^{\di}(\bfY) \neq \emptyset\Leftrightarrow\iball_{\bfe_2}^{\di}(\bfX) \cap \iball_{\bfe_1}^{\di}(\bfY) \neq \emptyset,
\end{align*}
and let $\bfD\in \dball_{\bfe_1}^{\di}(\bfX) \cap \dball_{\bfe_2}^{\di}(\bfY)$.
To apply~\cref{claim:dim-proj} and use the induction hypothesis, we project the \tensors\ on an axis different than the ones affected by a deletion. For the given case, we have $\di-2$ available axes to project on. Assume we project on the axis $x_\kappa$, where $\kappa \in [\di]\setminus\{1,2\}$. Thus, we transform the $(\di-1)$-dimensional hyperplane deletion in $\bfX$ and $\bfY$ to a $(\di-2)$-dimensional hyperplane deletion in $\proj{\kappa}(\bfX)$ and $\proj{\kappa}(\bfY)$ (c.f. ~\cref{claim:dim-proj}). Therefore, we can write that
\begin{align*}
&\dball_{\bfe_1}^{\di}(\bfX) \cap \dball_{\bfe_2}^{\di}(\bfY) \neq \emptyset,   \\ 
&\Leftrightarrow \dball_{\bfe_1}^{\di-1}(\proj{\kappa}(\bfX)) \cap \dball_{\bfe_2}^{\di-1}(\proj{\kappa}(\bfY)) \neq \emptyset, \\
 & \Leftrightarrow \iball_{\bfe_2}^{\di-1}(\proj{\kappa}(\bfX)) \cap \iball_{\bfe_1}^{\di-1}(\proj{\kappa}(\bfY)) \neq \emptyset,
\end{align*}
where the last equivalence follows from the induction hypothesis. Hence, there exists a $\proj{\kappa}(\bfI) \in \iball^{\di-1}_{\bfe_2}(\proj{\kappa}(\bfX)) \cap \iball^{\di-1}_{\bfe_1}(\proj{\kappa}(\bfY))$. Due to the fact that we have projected on an axis $x_\kappa \neq x_1, x_2$ and given \cref{claim:dim-proj}, we can interpret the $(\di-2)$-dimensional hyperplane insertion in $\proj{\kappa}(\bfX)$ and $\proj{\kappa}(\bfY)$ as a $(\di-1)$-dimensional hyperplane insertion in $\bfX$ and $\bfY$ by applying the inverse projection $\invproj{\kappa}(\cdot)$ to the projected \tensors. By the above observations we conclude that there exists a $\bfI \in \iball_{\bfe_2}^{\di}(\bfX) \cap \iball_{\bfe_1}^{\di}(\bfY)$ if there exists $\bfD\in \dball_{\bfe_1}^{\di}(\bfX) \cap \dball_{\bfe_2}^{\di}(\bfY)$ and conclude the ``if'' part of the proof.
\end{proof}

We now show the equivalence of $\bfone^{(d)}$-insertion and $\bfone^{(d)}$-deletion-correcting codes by using the results of~\cref{claim:dim-proj} and~\cref{lm:dd2}.
\begin{figure}[t]
  \centering
  \input{Figures/fig-1d-proof}
\caption{A flow chart of the proof of Theorem \ref{thm:eqinsdel3d} for $\di =3$. Given a \tensor\ $\bfD \in \dball^\di_{\bfone}(\bfX) \cap \dball^\di_{\bfone}(\bfY)$, we show the existence $\bfI \in \iball^\di_{\bfone}(\bfX) \cap \iball^\di_{\bfone}(\bfY)$. Given the existence of $\bfX$, $\bfY$, $\bfD$, and the orange \tensors\ one can show by \cref{lm:dd2} the existence of the green and purple marked \tensors\ and the \tensor\ $\bfI$. Since $\bfX, \bfY$ and $\bfI$ are connected via the purple \tensors\ as shown one can conclude the equivalence. \vspace{-1ex}}
  \label{fig:thmgen}
\end{figure}
\begin{theorem} \label{thm:eqinsdel3d} A code $\mathcal{C} \subseteq \Sigma_q^{n^{\odi}}$ is a $\bfone^{(d)}$-deletion-correcting code if and only if it is a $\bfone^{(d)}$-insertion-correcting code.
\end{theorem}
\begin{proof}
We provide an illustration of the proof for the case of $\di =3$ in \cref{fig:thmgen}. Assume there exists an \tensor\ $\bfD \in \Sigma_q^{(n-1)^{\odi}}$ such that $\bfD \in \dball^\di_{\bfone}(\bfX) \cap \dball^\di_{\bfone}(\bfY)$. For simplicity of notation, we fix the order of the deletions in $\bfX$ and $\bfY$ to obtain $\bfD$ to be an $x_\di$-deletion first, then an $x_{\di-1}$-deletion and so on until making an $x_1$-deletion. Note that the proof can be replicated for any ordering by the comprehensiveness of \cref{lm:dd2} which is our main building block. To prove the statement, we build a grid-like structure with axes $i,j \in [\di]$ and \tensors\ as grid points denoted by $\bfX^{i,j}$. We define $\bfX^{\di,0} \triangleq \bfX$, $\bfX^{0,\di} \triangleq \bfY$, and $\bfX^{0,0} \triangleq \bfD$. For fixed $j=0$, let the series of \tensors\ $\{\bfX^{i,0}\}_{i=0}^{\di}$ be defined such that $\bfX^{i-1,0} \in \dball^\di_{\bfe_i}(\bfX^{i,0})$ for $i=\{1,\dots,\di\}$. We define the series of \tensors\ $\{\bfX^{0,j}\}_{j=0}^{\di}$ similarly for fixed $i=0$. The strategy of the proof will show the existence of \tensors\ $\bfX^{i,j}$ for any $i,j \in [\di]$ such that $\bfX^{\di,\di} \in \iball^\di_{\bfone}(\bfX^{\di,0}) \cap \iball^\di_{\bfone}(\bfX^{0,\di})$. 

By the definition of the series we have that $\bfX^{0,0} \in \dball^\di_{\bfe_1}(\bfX^{1,0}) \cap \dball^\di_{\bfe_1}(\bfX^{0,1})$. By \cref{lm:dd2} there exists an \tensor\ $\bfX^{1,1} \in \iball^\di_{\bfe_1}(\bfX^{1,0}) \cap \iball^\di_{\bfe_1}(\bfX^{0,1})$. From that it follows that $\bfX^{1,0} \in \dball^\di_{\bfe_2}(\bfX^{2,0}) \cap \dball^\di_{\bfe_1}(\bfX^{1,1})$. By applying again \cref{lm:dd2} we have that there exists an \tensor\ $\bfX^{2,1} \in \iball^\di_{\bfe_1}(\bfX^{2,0}) \cap \iball^\di_{\bfe_2}(\bfX^{1,1})$. For $j=1$, by repeating the aforementioned strategy we can show the existence of the series of \tensors\ $\{\bfX^{i,1}\}_{i=2}^\di$. Given this series of \tensors\, one can show the existence $\{\bfX^{i,2}\}_{i=1}^\di$, where $j=2$ and given the starting statement $\bfX^{0,1} \in \dball^\di_{\bfe_1}(\bfX^{1,1}) \cap \dball^\di_{\bfe_2}(\bfX^{0,2})$. Therefore by consecutively incrementing $j \in [0,\di-1]$ and for each $j$ incrementing consecutively $i \in [0,\di-1]$, then for each pair $(i,j)$ by \cref{lm:dd2} one has the following equivalence: Given $\bfX^{i,j} \in \dball^\di_{\bfe_i}(\bfX^{i+1,j}) \cap \dball^\di_{\bfe_j}(\bfX^{i,j+1})$ there exists an \tensor\ $\bfX^{i+1,j+1} \in \iball^\di_{\bfe_j}(\bfX^{i+1,j}) \cap \iball^\di_{\bfe_i}(\bfX^{i,j+1})$. Therefore, we have proven the existence of an \tensor\ $\bfX^{\di,\di}\in \iball^\di_{\bfone}(\bfX^{\di,0}) \cap \iball^\di_{\bfone}(\bfX^{0,\di})$ which concludes the proof. 
\end{proof}

%% file: Figures/fig-1d-proof.tex
\newcommand{\LT}{L\ref{lm:dd2}}

\begin{tikzpicture}[
    arr/.style = {rectangle, rounded corners, draw=black,
                           minimum width=8ex, minimum height=2ex,
                           text centered, font=\tiny},
    givenbefore/.style = {arr, fill=color1!30},
    createdfirst/.style = {arr, fill=color2!30},
    giventhen/.style = {arr, fill=color6!30},
    createdsecond/.style = {arr, fill=color4!30},
    thm/.style = {circle, draw=black, fill=color3!30,
                           text centered, font=\tiny},
    arrowdescp/.style = {midway, fill=white, font=\footnotesize}]

    \def\xdist{7.0ex}
    \def\ydist{7.0ex}
    
    \definecolor{color0}{rgb}{0.12156862745098,0.466666666666667,0.705882352941177}
    \definecolor{color1}{rgb}{1,0.498039215686275,0.0549019607843137}
    \definecolor{color2}{rgb}{0.172549019607843,0.627450980392157,0.172549019607843}
    \definecolor{color3}{rgb}{0.83921568627451,0.152941176470588,0.156862745098039}
    \definecolor{color4}{rgb}{0.580392156862745,0.403921568627451,0.741176470588235}
    \definecolor{color6}{rgb}{0.549019607843137,0.337254901960784,0.294117647058824}
    \definecolor{color5}{rgb}{0.890196078431372,0.466666666666667,0.76078431372549}
    

    \node [arr,fill=color0!30] (x03) at (3*\xdist,3*\ydist) {$\bfY \triangleq \bfX^{0,3}$};
    \node [arr,fill=color0!30] (x30) at (-3*\xdist,3*\ydist) {$\bfX \triangleq \bfX^{3,0}$};
    
    \node [arr,fill=color0!30] (x00) at (0*\xdist,0*\ydist) {$\bfD \triangleq \bfX^{0,0}$};
    \node [arr,fill=color0!30] (x33) at (0*\xdist,6*\ydist) {$\bfI \triangleq \bfX^{3,3}$};
    
    \node [givenbefore] (x20) at (-2*\xdist,2*\ydist) {$\bfX^{2,0}$};
    \node [givenbefore] (x10) at (-1*\xdist,1*\ydist) {$\bfX^{1,0}$};
    \node [givenbefore] (x01) at (1*\xdist,1*\ydist) {$\bfX^{0,1}$};
    \node [givenbefore] (x02) at (2*\xdist,2*\ydist) {$\bfX^{0,2}$};
    
    \node [createdfirst] (x11) at (0*\xdist,2*\ydist) {$\bfX^{1,1}$};
    \node [createdfirst] (x21) at (-1*\xdist,3*\ydist) {$\bfX^{2,1}$};
    \node [createdfirst] (x12) at (1*\xdist,3*\ydist) {$\bfX^{1,2}$};
    \node [createdfirst] (x22) at (0*\xdist,4*\ydist) {$\bfX^{2,2}$};
    
    \node [createdsecond] (x31) at (-2*\xdist,4*\ydist) {$\bfX^{3,1}$};
    \node [createdsecond] (x32) at (-1*\xdist,5*\ydist) {$\bfX^{3,2}$};
    \node [createdsecond] (x13) at (2*\xdist,4*\ydist) {$\bfX^{1,3}$};
    \node [createdsecond] (x23) at (1*\xdist,5*\ydist) {$\bfX^{2,3}$};
    
    \draw [-] (x00) -- (x10) node [arrowdescp] {$\bfe_1$};
    \draw [-] (x01) -- (x11) node [arrowdescp] {$\bfe_1$};
    \draw [-] (x02) -- (x12) node [arrowdescp] {$\bfe_1$};
    \draw [-] (x03) -- (x13) node [arrowdescp] {$\bfe_1$};
    \draw [-] (x00) -- (x01) node [arrowdescp] {$\bfe_1$};
    \draw [-] (x10) -- (x11) node [arrowdescp] {$\bfe_1$};
    \draw [-] (x20) -- (x21) node [arrowdescp] {$\bfe_1$};
    \draw [-] (x30) -- (x31) node [arrowdescp] {$\bfe_1$};
    
    \draw [-] (x10) -- (x20) node [arrowdescp] {$\bfe_2$};
    \draw [-] (x11) -- (x21) node [arrowdescp] {$\bfe_2$};
    \draw [-] (x12) -- (x22) node [arrowdescp] {$\bfe_2$};
    \draw [-] (x13) -- (x23) node [arrowdescp] {$\bfe_2$};
    \draw [-] (x01) -- (x02) node [arrowdescp] {$\bfe_2$};
    \draw [-] (x11) -- (x12) node [arrowdescp] {$\bfe_2$};
    \draw [-] (x21) -- (x22) node [arrowdescp] {$\bfe_2$};
    \draw [-] (x31) -- (x32) node [arrowdescp] {$\bfe_2$};
    
    \draw [-] (x20) -- (x30) node [arrowdescp] {$\bfe_3$};
    \draw [-] (x21) -- (x31) node [arrowdescp] {$\bfe_3$};
    \draw [-] (x22) -- (x32) node [arrowdescp] {$\bfe_3$};
    \draw [-] (x23) -- (x33) node [arrowdescp] {$\bfe_3$};
    \draw [-] (x02) -- (x03) node [arrowdescp] {$\bfe_3$};
    \draw [-] (x12) -- (x13) node [arrowdescp] {$\bfe_3$};
    \draw [-] (x22) -- (x23) node [arrowdescp] {$\bfe_3$};
    \draw [-] (x32) -- (x33) node [arrowdescp] {$\bfe_3$};
    
    \foreach \x/\y in {0/1,0/3,0/5,-1/2,-1/4,-2/3,1/2,1/4,2/3}{
        \node [thm] (l00) at (\x*\xdist,\y*\ydist) {\LT};
    }

\end{tikzpicture}

%% file: t-general.tex
In this section we show the the equivalence of $\tcode^{(\di)}$-insertions and $\tcode^{(\di)}$-deletions in $\di$-dimensional \tensors\ for any number of $(\di-1)$-dimensional hyperplane insertions and deletions, respectively, i.e., we show the equivalence of $\tcode^{(\di)}$-deletion-correcting codes with  $\tcode^{(\di)}$-insertion-correcting codes for any $\tcode^{(\di)} \in \mathbb{Z}_{\geq 0}^d$. The proof follows similar steps as the one used by the authors in \cite{welter2021multiple} for the two-dimensional case.

\begin{theorem} \label{thm:eqins3dgen} A code $\mathcal{C}\subseteq \Sigma_q^{n^{\odi}}$ is a $\tcode^{(\di)}$-deletion-correcting code if and only if it is a $\tcode^{(\di)}$-insertion-correcting code.
\end{theorem}
\begin{proof}
For notational convenience we define the vector $\bfc^j \triangleq (c_1,\dots,c_{j-1},0,c_{j+1},\dots,c_{\di}) \in \mathbb{N}^{\di}$. 
In this proof, the vector $\tcode^{(\di)}$ can be written as $\tcode^{(\di)} = t_j\bfone^{(d)} + \bfc^j \triangleq \tcode_{j,c}$, where $t_j = \min_{i\in[\di]} t_i$ and $c_i\triangleq t_i-t_j$ for $i\in [\di]$, to emphasize the composition. Without loss of generality, we show the proof for $j=1$, since by symmetry the proof holds for all $j \in [\di]$, and write $t \triangleq t_1$. Let $t'\triangleq \sum_{i=1}^\di c_i$, the proof proceeds by induction over $t'$. 
For simplicity, we fix in some parts of the proof the order of the $x_i$-deletions. That serves for a better presentation of the proofs and incurs no loss of generality. 
In the proof the contraposition is shown, i.e., we show that $\dball^\di_{\tcode_{1,c}}(\bfX) \cap \dball^\di_{\tcode_{1,c}}(\bfY) \neq \emptyset$ if and only if $\iball^\di_{\tcode_{1,c}}(\bfX) \cap \iball^\di_{\tcode_{1,c}}(\bfY) \neq \emptyset$. We only show the ``if'' part since the ``only if'' part follows by using similar arguments.

\emph{Base case $\sum_{i=1}^{\di} c_i = 1$: } For the reader's convenience, a flowchart of the proof for $\di=3$ is presented in~\cref{fig:thmproofgen}. There are $\di-1$ possibilities for $\bfc^1$ such that $\sum_{i=1}^{\di} c_i = 1$. We show the proof steps for $\bfc^1 = \bfe_\kappa$, i.e., there is a combination of $t\bfone^{(d)}$-deletions and an extra $x_\kappa$-deletion for $\kappa \in [\di]$. 

 \begin{figure}[t]
  \centering
\input{Figures/fig-general-proof}
\caption{A flow chart of the proof of Theorem \ref{thm:eqins3dgen} for $\di =3$. Given an \tensor\ $\bfC_{k+1} \in \dball^\di_{\tcode_{1,c}}(\bfX) \cap \dball^\di_{\tcode_{1,c}}(\bfY)$, we show the existence $\bfG_{k+1} \in \iball^\di_{\tcode_{1,c}}(\bfX) \cap \iball^\di_{\tcode_{1,c}}(\bfY)$. Given the existence of $\bfX$, $\bfY$, $\bfC_{k+1}$ and the orange \tensors\ one can show by \cref{lm:dd2} the existence of the green marked \tensors\ and then by Theorem \ref{thm:eqinsdel3d2} and Lemma \ref{lm:dd2} the existence of brown and purple marked \tensors and the \tensor\ $\bfG_{k+1}$. }\vspace{-0.3cm}
  \label{fig:thmproofgen}
\end{figure}
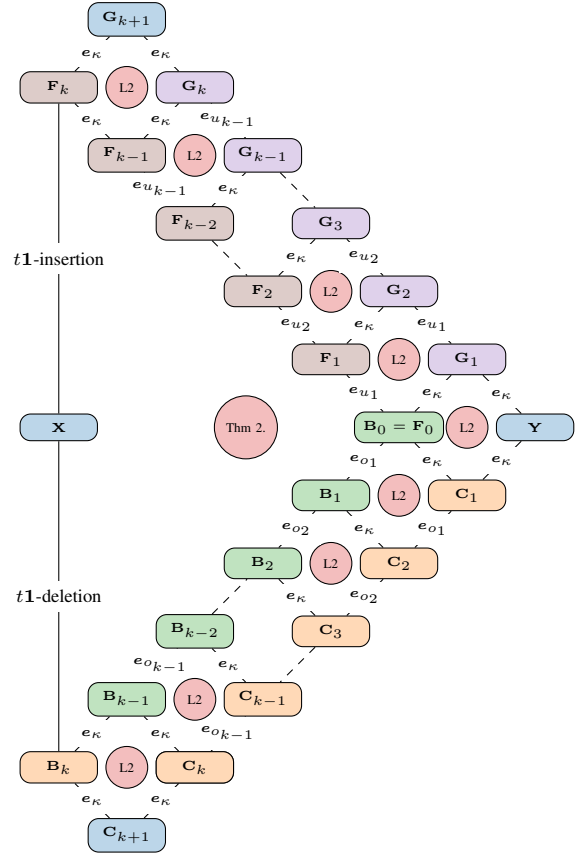 

For any two \tensors\ $\bfX,\bfY\in \Sigma_q^{n^{\odi}}$, assume there exists a \tensor\ $\bfD$ such that $\bfD \in \dball^\di_{\tcode_{1,c}}(\bfX) \cap \dball^\di_{\tcode_{1,c}}(\bfY)$. Let $k = \di t$, define the \tensor\ $\bfB_k$ such that $\bfB_k \in \dball^\di_{t\bfone}(\bfX)$ and $\bfD \in \dball^\di_{\bfe_\kappa}(\bfB_k)$ due to the choice of $\bfc^1$. For simplicity, we define $\bfC_1 \in \dball^\di_{\bfe_\kappa}(\bfY)$. Let $\bfo \in [\di]^{k}$ denote the vector whose entries $o_i$ denote the series of $x_{o_i}$-deletions to obtain $\bfD$ from $\bfC_1$ and fix $o_k = \kappa$. We define the series of \tensors\ $\{\bfC_s\}_{s=1}^{k+1}$ such that 
\begin{align*}
    \bfC_s \in \begin{cases} \dball^\di_{\bfe_{\kappa}}(\bfC_{s-1}) & \text{if $s = 1$ or $s = k+1$} \\
    \dball^\di_{\bfe_{o_{s-1}}}(\bfC_{s-1}) & \text{otherwise.}
    \end{cases} 
\end{align*}
where $\bfC_0 \triangleq \bfY$ and $\bfC_{k+1} \triangleq \bfD$. We show that there exists a series of \tensors\ $\{ \bfB_s\}_{s=0}^{k-1}$, resulting from hyperplane insertions starting from $\bfB_k$ and leading to an \tensor\ $\bfB_0 \in \Sigma_q^{n^{\odi}}$, such that $\bfB_k \in \dball^\di_{t\bfone}(\bfX) \cap \dball^\di_{t\bfone}(\bfB_{0})$. By the aforementioned definitions we have that $\bfC_{k+1} \in \dball^\di_{\bfe_\kappa}(\bfB_k) \cap \dball^\di_{\bfe_\kappa}(\bfC_k)$. By \cref{lm:dd2}, there exists a $\bfB_{k-1} \in \iball^\di_{\bfe_\kappa}(\bfB_k) \cap \iball^\di_{\bfe_\kappa}(\bfC_k)$. Applying~\cref{lm:dd2} sequentially shows the existence of the series of \tensors\ $\{\bfB_s\}_{s=0}^{k-1}$, i.e., by \cref{lm:dd2} for each $\bfC_{s+1} \in \dball^\di_{\bfe_\kappa}(\bfB_{s}) \cap \dball^\di_{\bfe_{o_s}}(\bfC_{s}) $ there exists a $\bfB_{s-1} \in \iball_{\bfe_{o_s}}(\bfB_{s}) \cap \iball^\di_{\bfe_\kappa}(\bfC_{s}) $ for $s \in \{k,\dots,1\}$.
Hence, we show the existence of an \tensor\ $\bfB_0 \in \Sigma_q^{n^{\odi}}$ such that $\bfB_k \in \dball^\di_{t\bfone}(\bfX) \cap \dball^\di_{t\bfone}(\bfB_0)$.

By~\cref{thm:eqinsdel3d2}, the existence of $\bfB_k$ implies the existence of an \tensor\ $\bfF_k \in \iball^\di_{t\bfone}(\bfX) \cap \iball^\di_{t\bfone}(\bfB_0)$, i.e., obtained by a $t\bfone^{(\di)}$-insertion in $\bfB_0$. Let $\insvec \in [\di]^{k}$ denote the vector whose entries $\insentry_i$ denote the series of $x_{\insentry_i}$-insertions to obtain $\bfF_k$ from $\bfB_0$ and fix $\insentry_k = \kappa$. We define the \tensors\ $\{\bfF_s\}_{s=1}^k$ such that
\begin{align*}
    \bfF_s \in \begin{cases} \iball^\di_{\bfe_{\kappa}}(\bfF_{s-1}) & \text{if $s = k$} \\
    \iball^\di_{\bfe_{\insentry_{s}}}(\bfF_{s-1}) & \text{otherwise,}
    \end{cases}
\end{align*}
where $\bfF_0 \triangleq \bfB_0$. Noting that $\bfC_1 \in \dball^\di_{\bfe_\kappa}(\bfB_0) \cap \dball^\di_{\bfe_\kappa}(\bfY)$ and applying~\cref{lm:dd2}, there exists an 
\tensor\ $\bfG_1 \in \iball^\di_{\bfe_\kappa}(\bfB_0) \cap \iball^\di_{\bfe_\kappa}(\bfY)$, which means that $\bfF_0 \in \dball^\di_{\bfe_{\insentry_1}}(\bfF_1) \cap \dball^\di_{\bfe_\kappa}(\bfG_1)$. By sequentially applying \cref{lm:dd2} we can show the existence of the series of \tensors\ $\{\bfG_s\}_{s=1}^{k+1}$ such that $\bfG_{k+1} \in \iball^\di_{\tcode_{1,c}}(\bfX) \cap \iball^\di_{\tcode_{1,c}}(\bfY)$. Meaning by the fact that $\bfF_{s-1} \in \dball^\di_{\bfe_{\insentry_s}}(\bfF_s) \cap \dball^\di_{\bfe_\kappa}(\bfG_s)$ there exists an \tensor\ $\bfG_{s+1} \in \iball^\di_{\bfe_\kappa}(\bfF_s) \cap \iball^\di_{\bfe_{\insentry_s}}(\bfG_s)$ for $s \in \{1,\dots,k\}$. 
Hence, we have shown that if there exists an \tensor\ $\bfC_{k+1} \in \dball^\di_{\tcode_{1,c}}(\bfX) \cap \dball^\di_{\tcode_{1,c}}(\bfY)$, then there exists an \tensor\ $\bfG_{k+1} \in \iball^\di_{\tcode_{1,c}}(\bfX) \cap \iball^\di_{\tcode_{1,c}}(\bfY)$, which concludes the base case.

\emph{Induction hypothesis: } Given any vector $\bfc^1$ such that 
$ \sum_{i=1}^{d} c_i = t^\prime$, and two \tensors\ $\bfX,\bfY \in \Sigma_q^{n^{\odi}}$ it holds that
\begin{align*}
    &\dball^d_{\tcode_{1,c}}(\bfX) \cap \dball^d_{\tcode_{1,c}}(\bfY) \neq \emptyset\;\Leftrightarrow \iball^d_{\tcode_{1,c}}(\bfX) \cap \iball^d_{\tcode_{1,c}}(\bfY) \neq \emptyset,
\end{align*}
where $\tcode_{1,c} \triangleq t\bfone+\bfc^1$.

\emph{Induction step: } Assume that the induction hypothesis holds for all values $0 \le \sum_{i=1}^{d} c_i = t^\prime$ where $c_1 = 0$. We prove that the hypothesis holds for $ \sum_{i=1}^{\di} c_i + 1 = t^\prime +1$, i.e., by adding an extra hyperplane deletion. Let the extra deletion be an $x_\kappa$-deletion and define $\tcode_{1,c}^\prime =  (t,t+c_2,\dots,t+c_\kappa+1,\dots,t+c_{\di})$.
Assume that there exists an \tensor\ $\bfD$ such that $\bfD \in \dball^\di_{\tcode_{1,c}^\prime}(\bfX) \cap \dball^\di_{\tcode_{1,c}^\prime}(\bfY)$. Let $k^\prime = \di t+t^\prime +1$, then we defined the \tensors\ $\bfB_{k^\prime }$ and $\bfC_{k^\prime }$ such that $\bfB_{k^\prime } \in \dball^d_{\tcode_{1,c}}(\bfX)$ and $\bfC_{k^\prime } \in \dball^d_{\tcode_{1,c}}(\bfY)$. The rest of the proof follows from the base case, by using $k^\prime$ instead of $k$ and therefore is omitted due to space limitations.
\end{proof}

By considering the collections of all $\tcode^{(\di)}$-deletion-correcting codes such that $\sum_{i=1}^{\di} t_i = t$ we have the following corollary.

\begin{corollary}
A code $\mathcal{C}\subseteq \Sigma_q^{n^{\odi}}$ is a $t^{(\di)}$-deletion-correcting code if and only if it is a $t^{(\di)}$-insertion-correcting code.
\end{corollary}

%% file: Figures/fig-general-proof.tex
   \def\xdist{5.72ex}
     \def\ydist{5.72ex}
\newcommand{\LT}{L\ref{lm:dd2}}
\newcommand{\LI}{L2}
\newcommand{\THRM}{Thm \ref{thm:eqinsdel3d2}.}

\begin{tikzpicture}[scale=1.0,
     arr/.style = {rectangle, rounded corners, draw=black,
                            minimum width=6.5ex, minimum height=0.9ex,
                            text centered, font=\tiny},
     givenbefore/.style = {arr, fill=color1!30},
     createdfirst/.style = {arr, fill=color2!30},
     giventhen/.style = {arr, fill=color6!30},
     createdsecond/.style = {arr, fill=color4!30},
     thm/.style = {circle, draw=black, fill=color3!30,
                            text centered, font=\tiny},
     arrowdescp/.style = {midway, fill=white, font=\tiny}]

\definecolor{color0}{rgb}{0.12156862745098,0.466666666666667,0.705882352941177}
\definecolor{color1}{rgb}{1,0.498039215686275,0.0549019607843137}
\definecolor{color2}{rgb}{0.172549019607843,0.627450980392157,0.172549019607843}
\definecolor{color3}{rgb}{0.83921568627451,0.152941176470588,0.156862745098039}
\definecolor{color4}{rgb}{0.580392156862745,0.403921568627451,0.741176470588235}
\definecolor{color6}{rgb}{0.549019607843137,0.337254901960784,0.294117647058824}
\definecolor{color5}{rgb}{0.890196078431372,0.466666666666667,0.76078431372549}

     \node [arr,fill=color0!30] (x1) at (-1*\xdist,0) {  $\bfX$};
     \node [arr,fill=color0!30] (g0) at (6*\xdist,0) {$\bfY$};
     \node [givenbefore] (ck) at (1*\xdist,-5*\ydist) {$\bfC_{k}$};
     \node [givenbefore] (ckm1) at (1*\xdist,-5*\ydist) {$\bfC_{k}$};
     \node [createdfirst] (bk1) at (0*\xdist,-4*\ydist) {$\bfB_{k-1}$};
      \node [givenbefore] (bk) at (-1*\xdist,-5*\ydist) {$\bfB_{k}$};

     \node [givenbefore] (ckm2) at (2*\xdist,-4*\ydist) {$\bfC_{k-1}$};
      \node [givenbefore] (c3) at (3*\xdist,-3*\ydist) {$\bfC_{3}$};
     \node [givenbefore] (c2) at (4*\xdist,-2*\ydist) {$\bfC_{2}$};
     \node [givenbefore] (c1) at (5*\xdist,-1*\ydist) {$\bfC_{1}$};
     
     \node [arr,fill=color4!30] (gkp1) at (1*\xdist,5*\ydist)
{$\bfG_{k}$};
     \node [arr,fill=color0!30] (gkp2) at (0*\xdist,6*\ydist)
{$\bfG_{k+1}$};
     \node [arr,fill=color0!30,] (ckp1) at (0*\xdist,-6*\ydist)
{$\bfC_{k+1}$};
     \draw [-] (g0) -- (c1) node [arrowdescp] {$\bfe_\kappa$};
     \draw [-] (c1) -- (c2) node [arrowdescp] {$\bfe_{o_1}$};
     \draw [ -] (ckm1) -- (ckp1) node [arrowdescp] {$\bfe_\kappa$};
     \draw [-] (bk) -- (ckp1) node [arrowdescp] {$\bfe_\kappa$};
     \draw [-] (c3) -- (c2) node [arrowdescp] {$\bfe_{o_2}$};
     \draw [dashed] (c3) -- (ckm2);
     \draw [-] (ckm1) -- (ckm2) node [arrowdescp] {$\bfe_{o_{k-1}}$};
     \draw [-] (x1) -- (bk) node [arrowdescp] {\scriptsize $t\bfone$-deletion};

     \node [createdfirst] (bkm1) at (1*\xdist,-3*\ydist) {$\bfB_{k-2}$};
     \node [createdfirst] (b1) at (3*\xdist,-1*\ydist) {$\bfB_{1}$};
     \node [createdfirst] (bkm2) at (2*\xdist,-2*\ydist) {$\bfB_{2}$};
      \node [createdfirst] (f0) at (4*\xdist,0) {$\bfB_0 = \bfF_0$};
      
     \draw [-] (bkm2) -- (c3) node [arrowdescp] {$\bfe_\kappa$};

     \draw [-] (bk1) -- (bk) node [arrowdescp] {$\bfe_\kappa$};
     \draw [-] (bkm1) -- (bk1) node [arrowdescp] {$\bfe_{o_{k-1}}$};
     \draw [-] (bk1) -- (ckm1) node [arrowdescp] {$\bfe_\kappa$};
     \node [thm] (lk) at (1*\xdist,-4*\ydist) {\LT};

     \draw [ dashed] (bkm1) -- (bkm2);
     \draw [-] (bkm1) -- (ckm2) node [arrowdescp] {$\bfe_\kappa$};
      \node [thm] (lkm1) at (3*\xdist,-2*\ydist) {\LI};
     \node [thm] (lkm1) at (0*\xdist,-5*\ydist) {\LI};

        \draw [dashed] (bkm2) -- (b1) node [arrowdescp] {$\bfe_{o_2}$};
        \draw [ -] (f0) -- (b1) node [arrowdescp] {$\bfe_{o_1}$};
        \draw [ -] (b1) -- (c2) node [arrowdescp] {$\bfe_\kappa$};
        \node [thm] (l1) at (4*\xdist,-1*\ydist) {\LI};
        \draw [ -] (f0) -- (c1) node [arrowdescp] {$\bfe_\kappa$};

         \node [giventhen] (fk2) at (-1*\xdist,5*\ydist) {$\bfF_k$};

       \node [giventhen] (fk) at (0*\xdist,4*\ydist) {$\bfF_{k-1}$};
     \draw [ -] (x1) -- (fk2) node [arrowdescp] {\scriptsize $t\bfone$-insertion};
     \node [thm] (lk) at (1.75*\xdist,0*\ydist) {\THRM};

     \node [giventhen] (fkm1) at (1*\xdist,3*\ydist) {$\bfF_{k-2}$};
     \node [giventhen] (fkm2) at (2*\xdist,2*\ydist) {$\bfF_{2}$};
     \node [giventhen] (f1) at (3*\xdist,1*\ydist) {$\bfF_{1}$};
     \draw [ -] (fk) -- (fkm1) node [arrowdescp] {$\bfe_{\insentry_{k-1}}$};
	 \draw [-] (f1) -- (f0) node [arrowdescp] {$\bfe_{\insentry_1}$};
     \draw [dashed] (fkm1) -- (fkm2);
     \draw [dashed] (fkm2) -- (f1) node [arrowdescp] {$\bfe_{\insentry_2}$};

    	 \node [createdsecond] (g1) at (5*\xdist,1*\ydist) {$\bfG_{1}$};
     \node [createdsecond] (gk) at (2*\xdist,4*\ydist) {$\bfG_{k-1}$};
     \node [createdsecond] (gkm1) at (3*\xdist,3*\ydist) {$\bfG_3$};
     \node [createdsecond] (g2) at (4*\xdist,2*\ydist) {$\bfG_{2}$};

	 \draw [ -] (f0) -- (g1) node [arrowdescp] {$\bfe_\kappa$};
     \draw [ -] (g0) -- (g1) node [arrowdescp] {$\bfe_\kappa$};
     \node [thm] (ul0) at (5*\xdist,0*\ydist) {\LT};

          \draw [-] (fk2) -- (fk) node [arrowdescp] {$\bfe_\kappa$};
      \draw [-] (gkp1) -- (gkp2) node [arrowdescp] {$\bfe_\kappa$};
          \draw [-] (fk2) -- (gkp2) node [arrowdescp] {$\bfe_\kappa$};

      \draw [ -] (f1) -- (g2) node [arrowdescp] {$\bfe_\kappa$};
     \draw [ -] (g1) -- (g2) node [arrowdescp] {$\bfe_{\insentry_1}$};
 		\node [thm] (u1) at (4*\xdist,1*\ydist) {\LI};

     \draw [dashed] (gk) -- (gkm1);
     \draw [ -] (fkm2) -- (gkm1) node [arrowdescp] {$\bfe_\kappa$};
     \node [thm] (ukm1) at (3*\xdist,2*\ydist) {\LI};
     \draw [-] (gk) -- (fkm1) node [arrowdescp] {$\bfe_\kappa$};

     \draw [dashed] (gkm1) -- (g2) node [arrowdescp] {$\bfe_{\insentry_2}$};

     \draw [ -] (fk) -- (gkp1) node [arrowdescp] {$\bfe_\kappa$};

     \draw [ -] (gk) -- (gkp1) node [arrowdescp] {$\bfe_{\insentry_{k-1}}$};
     \node [thm] (uk) at (1*\xdist,4*\ydist) {\LT};
    \node [thm] (uk) at (0*\xdist,5*\ydist) {\LT};

\end{tikzpicture}

%% file: insdel-equiv.tex
So far we have only considered the equivalence between insertion and deletion correcting codes. In this section we are going to discuss the equivalence between $\tcode^{(\di)}$-deletion and $\tcode^{(\di)}$-insdel correcting codes. First, we need the following claim.
\begin{claim}\label{cl:insdel-vec}
For positive integers $m_1,\dots,m_\di$, $i \in [\di]$, a vector $\bfr^i = (0,\dots, 0,r_i,0, \dots,0)$, and any two \tensors\ $\bfX, \bfY \in \Sigma_q^{\cart_{\ell=1}^\di m_\ell}$ it holds that
\begin{align*}
\dball_{\bfr^i}^\di(\bfX) \cap \dball_{\bfr^i}^\di(\bfY) \neq \emptyset \Leftrightarrow \idball_{\bfr^i}^\di(\bfX) \cap \idball_{\bfr^i}^\di(\bfY) \neq \emptyset. 
\end{align*}
\end{claim}
\begin{proof}
We only show the ``if'' part, since the ``only if'' part follows by similar arguments. Let $\bfD \in \dball_{\bfr^i}^\di(\bfX) \cap \dball_{\bfr^i}^\di(\bfY)$. We define a consecutive series of projections of an \tensor\ $\bfX$ along the axes in a set $\cI \subseteq [d]$ by $\proj{\cI}(\bfX)$. Let $\cI = [d] \setminus \{i\}$, we have $\proj{\cI}(\bfX), \proj{\cI}(\bfY) \in \Sigma_{q^{n(d-1)}}^n$. Since we do not project along the axis affected by deletions we can transform the $(d-1)$-hyperplane deletions to symbol deletions in $\proj{\cI}(\bfX), \proj{\cI}(\bfY)$ by Claim \ref{claim:dim-proj}. Thus, there exits a $\proj{\cI}(\bfD) \in \dball_{r_i}^1(\proj{\cI}(\bfX)) \cap \dball_{r_i}^1(\proj{\cI}(\bfY))$ such that $\invproj{\cI}(\proj{\cI}(\bfD)) = \bfD$. Hence, by \cite{Cullina2014} there exists a $\proj{\cI}(\bfI) \in \idball_{r_i}^1(\proj{\cI}(\bfX)) \cap \idball_{r_i}^1(\proj{\cI}(\bfY))$. According to Claim \ref{claim:dim-proj} it follows that there exists a $\invproj{\cI}(\proj{\cI}(\bfI)) = \bfI \in \idball_{\bfr^i}^\di(\bfX) \cap \idball_{\bfr^i}^\di(\bfY)$, since all entries of $\bfr^i$ are zero except the $i$-th position.
\end{proof}
It is important to note that the position of $r_i$ within the vector $\bfr^i$ must remain the same for any equivalence. This means that $x_i$-deletions are only equivalent to $x_i$-insdels and not to $x_j$-insdels, $j\neq i$. We show this idea through a counterexample for two-dimensional arrays.
\begin{counterexample} 
The equivalence of a $(1,0)$-deletion-correcting code and a $(0,1)$-deletion-correcting code does not hold. To show this, we consider two arrays $\bfX, \bfY \in \Sigma^{3\times 3}$ and assume there exists an array $\bfD \in \Sigma^{2\times 3}$ such that $\bfD \in \dball_{1,0}^2(\bfX)\cap\dball_{1,0}^2(\bfY)$ as follows.
\begin{align*}
\bfX = \begin{pmatrix}
1 & 1 & 1 \\
0 & 1 & 0 \\
0 & 1 & 1 
\end{pmatrix},\;\; \bfY = \begin{pmatrix}
1 & 0 & 1 \\
0 & 1 & 0 \\
0 & 0 & 1 
\end{pmatrix},\;\;\bfD = \begin{pmatrix}
1 &  1 \\
0 &  0 \\
0 &  1 
\end{pmatrix},
\end{align*}
where $\bfD$ is obtained by deleting the second column from $\bfX$ and $\bfY$. Since more than one row of $\bfX$ and $\bfY$ are different, we see that  $\dball_{0,1}^2(\bfX)\cap\dball_{0,1}^2(\bfY) = \emptyset$ and therefore the equivalence does not hold.
\end{counterexample}

Given this result, we show that the insertion/deletion equivalence holds if one fixes a number of insdel for each dimension to be deleted.
\begin{lemma}
For positive integers $m_1,\dots,m_\di$, $i \in [\di]$, a vector $\bft=(t_1,\dots,t_\di) \in \mathbb{N}^\di$, and any two arrays $\bfX, \bfY \in \Sigma_q^{\cart_{\ell=1}^\di m_\ell}$ it holds that,
\begin{align*}
\dball_{\bft}^\di(\bfX) \cap \dball_{\bft}^\di(\bfY) \neq \emptyset \Leftrightarrow \idball_{\bft}^\di(\bfX) \cap \idball_{\bft}^\di(\bfY) \neq \emptyset. 
\end{align*}
\end{lemma}
\begin{proof}
We only show the ``only if'' part, since the ``if'' part follows by similar arguments. Let $\bft^{\mathrm{ins}} = (t_1^{\mathrm{ins}},t_2^{\mathrm{ins}},\dots,t_\di^{\mathrm{ins}})$ and $\bft^{\mathrm{del}} = (t_1^{\mathrm{del}},t_2^{\mathrm{del}},\dots,t_\di^{\mathrm{del}})$ such that $\bft = \bft^{\mathrm{ins}} + \bft^{\mathrm{del}}$. Assume that there exists an \tensor\ $\bfI \in \Sigma_q^{\cart_{i=1}^\di (m_i+(t_i^{\mathrm{ins}}-t_i^{\mathrm{del}}))} $ such that $\bfI \in \idball_{\bft}^\di(\bfX) \cap \idball_{\bft}^\di(\bfY)$. The order of deletions and insertions matters here, therefore we define $\bfX^\prime$ and $\bfY^\prime$ to be the \tensors\ resulting from $\bft^{\mathrm{del}}$-deletion, i.e., it holds that $\bfX^\prime \in \dball_{\bft^{\mathrm{del}}}^\di(\bfX)$ and $\bfY^\prime \in \dball_{\bft^{\mathrm{del}}}^\di(\bfY)$.  It then follows that $\bfI \in \iball_{\bft^{\mathrm{ins}}}^\di(\bfX^\prime)\cap\iball_{\bft^{\mathrm{ins}}}^\di(\bfY^\prime) $. By Theorem~\ref{thm:eqins3dgen}, there exists an \tensor\ $\bfD\in  \Sigma_q^{\cart_{i=1}^\di (m_i-(t_i^{\mathrm{ins}}+t_i^{\mathrm{del}}))} $ such that $\bfD \in \dball_{\bft^{\mathrm{ins}}}^\di(\bfX^\prime)\cap\dball_{\bft^{\mathrm{ins}}}^\di(\bfY^\prime) $ and as a result $\bfD \in \dball_{\bft}^\di(\bfX)\cap\dball_{\bft}^\di(\bfY) $.
\end{proof}

%% file: appendix.tex
In this section, we provide a proof of Theorem~\ref{thm:eqinsdel3d2}, i.e., we prove that a code $\mathcal{C} \subseteq \Sigma_q^{n^{\odi}}$ is a $t\bfone^{(d)}$-deletion-correcting code if and only if it is a $t\bfone^{(d)}$-insertion-correcting code.

The proof requires the following intermediate results.
\begin{claim} \label{claim:cor1} For any two \tensors\ $\bfX_1,\bfX_{t+1}\in \Sigma_q^{n^{\odi}}$, 
$\dball^\di_{t\bfone}(\bfX_1) \cap \dball^\di_{t\bfone}(\bfX_{t+1}) \neq \emptyset$ if and only if there exist $t-1$ \tensors\ $\bfX_2,\dots,\bfX_{t} \in  \Sigma_q^{n^{\odi}}$ such that $\dball^\di_{\bfone}(\bfX_i) \cap \dball^\di_{\bfone}(\bfX_{i+1}) \neq \emptyset$ for all $1 \le i \le t$.
\end{claim}
\begin{proof} 
We prove the ``if'' part by induction over $t$. The proof for the ``only if'' part follows similarly and is omitted. First, we define the base case of the induction, then the induction hypothesis and finally the induction step.
 
\emph{Base case $t=1$: } This is a trivial case in which the statement is already satisfied, i.e., there are no intermediate \tensors\ since $\dball^\di_{\bfone}(\bfX_1) \cap \dball^\di_{\bfone}(\bfX_{2}) \neq \emptyset$.

\emph{Induction hypothesis: } Assume that the statement holds for a given $t\in[n-2]$. That is, there exist two \tensors\ $\bfX_1,\bfX_{t+1}\in \Sigma_q^{n^{\odi}}$ that satisfy $\dball^\di_{t\bfone}(\bfX_1) \cap \dball^\di_{t\bfone}(\bfX_{t+1}) \neq \emptyset$ and there exist $t-1$ \tensors\ $\bfX_2,\dots,\bfX_{t}\in \Sigma_q^{n^{\odi}}$ such that $\dball^\di_{\bfone}(\bfX_i) \cap \dball^\di_{\bfone}(\bfX_{i+1}) \neq \emptyset$ for all $1 \le i \le t$.

\emph{Induction step: } We show that the statement holds for $t+1$. Let $\bfX_1,\bfX_{t+2}\in \Sigma_q^{n^{\odi}}$ be such that $\dball^\di_{(t+1)\bfone}(\bfX_1) \cap \dball^\di_{(t+1)\bfone}(\bfX_{t+2}) \neq \emptyset$. Define the two \tensors\ $\bfXov_{1},\bfXov_{t+1} \in \Sigma_q^{(n-1)^{\odi}}$ that result from $\bfone^{(\di)}$-deletion form $\bfX_1$ and $\bfX_{t+2}$ respectively, i.e., $\bfXov_{1} \in \dball^\di_{\bfone}(\bfX_1) $ and $\bfXov_{t+1} \in \dball^\di_{\bfone}(\bfX_{t+2})$ and  $\dball^\di_{t\bfone}(\bfXov_{1}) \cap \dball^\di_{t\bfone}(\bfXov_{t+1}) \neq \emptyset$ . Then, by the induction hypothesis, there exist $t-1$ \tensors\ $\bfXov_2,\dots,\bfXov_{t} \in \Sigma_q^{(n-1)^{\odi}} $ such that $\dball^\di_{\bfone}(\bfXov_i) \cap \dball^\di_{\bfone}(\bfXov_{i+1}) \neq \emptyset$ for all $1 \le i \le t$.

Then, by Theorem~\ref{thm:eqinsdel3d} we deduce that $\iball^\di_{\bfone}(\bfXov_i) \cap \iball^\di_{\bfone}(\bfXov_{i+1}) \neq \emptyset$ for $1 \le i \le t$, therefore, there exist $t$ \tensors\ $\bfX_2,\dots,\bfX_{t+1} \in \Sigma_q^{n^{\odi}}$ such that for all $2 \le i \le t+1$, it holds that $\bfX_i \in \iball^\di_{\bfone}(\bfXov_{i-1}) \cap \iball^\di_{\bfone}(\bfXov_{i})$. By definition, $\bfX_1 \in \iball^\di_{\bfone}(\bfXov_{1}) $ and $ \bfX_{t+2} \in \iball^\di_{\bfone}(\bfXov_{t+1})$, combine with $\bfX_2 \in \iball^\di_{\bfone}(\bfXov_{1})$ and $\bfX_{t+1} \in \iball^\di_{\bfone}(\bfXov_{t+1})$ derived from the aforementioned result, it holds that $\bfXov_{1} \in \dball^\di_{\bfone}(\bfX_1) \cap \dball^\di_{\bfone}(\bfX_{2})$ and $\bfXov_{t+1} \in \dball^\di_{\bfone}(\bfX_{t+1}) \cap \dball^\di_{\bfone}(\bfX_{t+2})$. Consequently, we showed that for $1 \le i \le t+1$ it holds that,
\begin{align}
    \bfXov_{i} \in \dball^\di_{\bfone}(\bfX_i) \cap \dball^\di_{\bfone}(\bfX_{i+1}) .
\end{align}
This completes the ``if'' part of the proof.
\end{proof}
\noindent Next we state a similar result for the insertion case.
\begin{claim}\label{claim:cor2} For any two \tensors\ $\bfX_1,\bfX_{t+1}\in \Sigma_q^{n^{\odi}}$, $\iball^\di_{t\bfone}(\bfX_1) \cap \iball^\di_{t\bfone}(\bfX_{t+1}) \neq \emptyset$ if and only if there exist $t-1$ \tensors\ $\bfX_2,\dots,\bfX_{t}$ such that $\iball^\di_{\bfone}(\bfX_i) \cap \iball^\di_{\bfone}(\bfX_{i+1}) \neq \emptyset$ for all $1 \le i \le t$.
\end{claim}
\begin{proof} Follows using similar statements as in Claim \ref{claim:cor1}.
\end{proof}
Theorem~\ref{thm:eqinsdel3d2} can be proven using the results of Claim~\ref{claim:cor1} and Claim~\ref{claim:cor2} as follows. For any two \tensors\ $\bfX_1,\bfX_{t+1}\in \Sigma_q^{n^{\odi}}$, if  $\dball^\di_{t\bfone}(\bfX_1) \cap \dball^\di_{t\bfone}(\bfX_{t+1}) \neq \emptyset$ then form Claim~\ref{claim:cor1} we know that there exist $t-1$ \tensors\ $\bfX_2,\dots,\bfX_{t}\in \Sigma_q^{n^{\odi}}$ such that $\dball^\di_{1}(\bfX_i) \cap \dball^\di_{1}(\bfX_{i+1}) \neq \emptyset$ for all $1 \le i \le t$. According to Theorem~\ref{thm:eqinsdel3d}, there exist $t$ arrays $\bfXov_1,\dots,\bfXov_{t} \in \Sigma_q^{(n+1)^{\odi}} \in \Sigma_q^{(n^{\odi}}$ such that for all $1 \le i \le t$,
\begin{align}
        \bfXov_{i} \in \iball^\di_{\bfone}(\bfX_i) \cap \iball^\di_{\bfone}(\bfX_{i+1}) .
\end{align}
Finally, by applying Claim~\ref{claim:cor2} we conclude that $\iball^\di_{t\bfone}(\bfX_1) \cap \iball^\di_{t\bfone}(\bfX_{t+1}) \neq \emptyset $. The ``only if'' part follows similarly.